\documentclass[runningheads]{llncs}
\usepackage{graphicx}
\usepackage{graphicx}
\usepackage{algorithm}
\usepackage{subcaption}
\usepackage{mwe}
\usepackage{breqn}

\usepackage{amsthm}
\usepackage[noend]{algpseudocode}
\usepackage{amsmath}
\usepackage{amssymb}
\usepackage[shortlabels]{enumitem}
\usepackage{dsfont}
\usepackage{tikz-network}

\newcommand{\mA}{\mathcal{A}}

\newcommand{\mS}{\mathcal{S}}

\newcommand{\mT}{\mathcal{T}}

\newcommand{\mG}{G}

\newcommand{\mP}{\mathcal{P}}
\newcommand{\mF}{\mathcal{F}}
\newcommand{\bbone}{\mathds{1}}

{
\newtheorem{assumption}{Assumption}
}
\begin{document}
\title{Synthesis of Static Test Environments for Observing Sequence-like Behaviors in Autonomous Systems}
\titlerunning{Formal Test Environment Synthesis}
%
\author{Apurva Badithela \and
Richard M. Murray}
\authorrunning{A. Badithela and R.M. Murray}
%
\institute{California Institute of Technology, Pasadena CA 91106, USA
\\\email{\{apurva, murray\}@caltech.edu}}
\maketitle              
\begin{abstract}
In this paper, we investigate formal test-case generation for high-level mission objectives, specifically reachability, of autonomous systems. We use Kripke structures to represent the high-level decision-making of the agent under test and the abstraction of the test environment. First, we define the notion of a test specification, focusing on a fragment of linear temporal logic represented by sequence temporal logic formulas. Second, we formulate the problem of test graph synthesis to find a test configuration for which the agent must satisfy the test specification to satisfy its mission objectives. We an algorithm, based on network flows, for synthesizing a test graph by restricting transitions, represented by edge deletions, on the original graph induced by the Kripke structures. The algorithm synthesizes the test graph iteratively using an integer linear program. We prove completeness for our algorithm, and we show that the edge deletions in each iteration maintain feasibility of the integer linear program in the subsequent iteration. We formalize the notion of a minimally constrained test graph in terms of maximum flow, and prove the synthesized test graph to be minimally constrained. We demonstrate our algorithm on a simple graph and on gridworlds.
\end{abstract}

\section{Introduction}
The 2007 DARPA Urban Challenge ushered interest in autonomous driving in urban environments~\cite{DARPA_Urban}. Participating vehicles had to pass three small-scale operational test-courses designed to evaluate the autonomous car's ability to satisfy safety, basic and advanced navigation requirements, and basic and advanced traffic scenarios~\cite{DARPA_Urban_rules}. Exhaustive verification for such complex safety-critical systems is prohibitive, creating a need for a formal operational testing framework to certify reliability of these systems~\cite{koopman2016challenges}. 
\par
Due to robustness metrics from their quantitative semantics, signal temporal logic (STL) and metric temporal logic (MTL), are natural paradigms for reasoning over trajectories of low-level continuous dynamics~\cite{fainekos2009robustness,donze2010robust}. In many instances, the term testing is used inter-changeably with falsification~\cite{sankaranarayanan2012falsification}. Falsification is the problem of finding initial conditions and input signals that lead to violation of a temporal logic formula with the goal of finding such failures quickly and for black-box models~\cite{annpureddy2011s,plaku2013falsification,donze2010breach,chou2018using}. Furthermore, the black-box approaches in the related topics of falsification of hybrid systems~\cite{annpureddy2011s}, and simulation-based test generation~\cite{tuncali2018simulation,fainekos2012verification}, rely on stochastic optimization algorithms to minimize the robustness of temporal logic satisfaction. Since dense-time temporal logics better encapsulate the range of system behaviors at the with continuous dynamics, these techniques are successful at falsification at the low-level. However, some of the complexity can be attributed to the coupling between continuous dynamics with high-level discrete decision-making behaviors, a hierarchical approach to test-case generation could be effective. 
\par 
At higher levels of abstraction comprising of discrete states, LTL specifications are often used to capture mission objectives. Covering arrays have been used to initialize discrete parameters of the test configuration at the start of the falsification procedure in~\cite{annpureddy2011s,tuncali2018simulation,majumdar2019paracosm}. In this work, we focus on a framework for testing of high-level specifications in linear temporal logic (LTL). 
 Here we adopt a different notion of testing -- one that is focused on observing the autonomous agent undertake a certain behavior in its mission. The DARPA Urban Challenge test courses, that mainly comprised of static obstacles and (dynamic) human-driven cars, were carefully designed to observe the agent undertaking certain behaviors~\cite{DARPA_Urban_rules}. For example, a part of the test course was designed for assessing parking behavior. The static obstacles -- barriers blocking the region in front of the parking lot and other parked cars -- were placed such that the agent had to repeatedly reverse/pull-in to incrementally adjust its heading angle before successfully parking in the designated spot. The clever placement of static obstacles in this scenario made it a challenging test for the agent, as opposed to an environment in which the agent pulls-in straight into the parking spot. Similarly, carefully designed scenarios with human-driven cars sought to observe other behaviors of the agent. In many, but not necessarily all, of these scenarios, the high-level behavior of the agent can be described as a sequence of waypoints. In the parking lot example, the sequence of waypoints can be characterized as a sequence of agent states, which can be characterized as a product of position and heading angle in the high-level abstraction. As a step towards automatically synthesizing these test scenarios, this paper asks the following question: \medskip
 \newline
\textit{Problem (Informal): Given a valid, user-defined sequence of waypoints, a reachability objective for the mission specification, find a set of possible initial conditions for the agent (if not specified by user) and determine a set of static constraints, characterized by transitions that are blocked/restricted, such that:
\begin{enumerate}[-,topsep=0pt]
    \item[i)]  the agent must visit the sequence of waypoints in order before its goal, and
    \item[ii)] the test environment is minimally constrained.
\end{enumerate}
}
\medskip
Most formal analyses of discrete state systems and temporal logic specifications are based on the model-checking paradigm of constructing a product automaton of the system abstraction and the B\"uchi automaton of the temporal logic formula~\cite{baier2008principles}. While this is a powerful approach that encapsulates any temporal logic formula, the B\"uchi automaton has potentially exponential number of states in the length of the formula. However, when restricted to a specific class of temporal logic formulas, it is possible to develop an alternative approach to analyzing the properties of a discrete-state system. In the paper, we use the concept of network flows~\cite{cormen2009introduction} and integer linear programming (ILP) to synthesize ``cuts'' or static constraints to characterize the test environment. In particular, the waypoints and goal state constitute the sources and sinks of these flows, and we seek to enforce static constraints that maximize the flow between consecutive waypoints such that the overall flow from the initial waypoint to the final goal is maximized. 

The contributions of this paper are as follows:
\begin{enumerate}
    \item Formulating the static test synthesis problem, formalizing the notions of a test specification and a test graph,
    \item An algorithm and an ILP subroutine that iteratively synthesizes static constraints using fewer variables in each iteration,
    \item Proof of completeness and recursive feasibility (feasibility of the ILP constraint set is invariant despite static constraints determined in prior iterations) of the aforementioned iterative algorithm,
    \item Formalizing the notion of a minimally constrained test graph, and a proof showing that the synthesized test environment is minimally constrained.
\end{enumerate}
The rest of this paper is organized as follows: In Section~\ref{sec:Prelims}, we define the test specification and the test graph, and formally present the static test environment synthesis problem. The baseline ILP and the iterative algorithm to solve the static test synthesis problem are detailed in Section~\ref{sec:Algorithms}, and subsequently, we prove recursive feasibility, completeness, and show that the synthesized static environment is minimally constrained. In Section~\ref{sec: Examples}, we illustrate the algorithms on parametrized gridworlds. Finally, we present future directions in Section~\ref{sec:Conclusions}.
\section{Preliminaries and Problem Statement}
\label{sec:Prelims}
In this section, we first introduce the notion of discrete finite state systems. Afterwards, we formally describe the specification that such system should satisfy. Finally, we present basic definitions from flow networks, which we will leverage later on for test synthesis.
\begin{definition}[Kripke Structure~\cite{baier2008principles}]
A \emph{Kripke structure} is a tuple \(\mT = (Q, A, \Delta, Q_0, AP, L\)), where \(Q\) is a finite set of states, \(A\) is a finite set of actions, \(\Delta: Q\rightarrow 2^{Q}\) is a non-deterministic transition function, \(Q_0\) is the set of initial states, and \(\mathcal{AP}\) is the set of atomic propositions, and \(L: Q \rightarrow 2^{\mathcal{AP}}\) is a labeling function. A directed graph \(G = (V, E)\) can be induced from \(\mT\) in which the vertices represent states \(Q\) and the edges represent the transitions, and the labeling function assigns propositions that are true at each vertex. For a proposition \(p\), and vertex \(v\), \(v \vdash p\) means that \(p\) is satisfied at \(v\). A \emph{run} \(\sigma = q_0 q_1 \ldots\) on the graph is an infinite sequence of its nodes where \(q_i \in Q\) represents the state at index \(i\).
\end{definition}
Linear Temporal Logic (LTL) is a formal specification language, originally developed for software systems, that has also been used for characterizing high-level specifications of cyber-physical systems~\cite{emerson1990temporal,manna2012temporal,alur2015principles,kress2009temporal}. For the syntax and semantics of LTL, refer to Appendix or ~\cite{baier2008principles}.
Finally, we present a few definitions on network flows that will become relevant for the purpose of synthesizing the test environment. 
\begin{definition}[Flows~\cite{cormen2009introduction}]
A flow network is a directed graph \(G=(V,E)\) with a non-negative capacity function \(c : E\rightarrow \mathbb{R}\) such that each edge \((u,v) \in E\) is assigned a non-negative capacity \(c(u,v)\geq 0\). The flow network also comprises of a source \(s \in V\) and the sink \(t \in V\). A flow is a function \(f:V\times V \rightarrow \mathbb{R}\) that satisfies the capacity constraint and flow conservation properties (see ~\cite{cormen2009introduction}). The value of the flow is given as follows,
\begin{equation*}
    f_G(s,t) = \sum_{v\in V}f(s,v) - \sum_{v\in V}f(v,s).
\end{equation*}
See Appendix or~\cite{cormen2009introduction} for a definition of residual flow networks.
\end{definition}
\begin{definition}[Augmenting Paths~\cite{cormen2009introduction}]
An augmenting path \(P\) is a simple path, that is, each edge is traversed once, from source \(s\) to sink \(t\) on the residual network \(G_f\). Maximizing the flow value \(f_G(s,t)\) leaves the residual network \(G_f\) with no augmenting paths. We say that a set of augmenting paths \(S = \{P_1,\ldots, P_{f_G(s,t)}\}\) represents an instance of the maximum flow from \(s\) to \(t\) in \(G\) if the flow network \(G_f\) is constructed by identifying the augmenting paths in \(S\). The set \(S\) is said to be the set of shortest augmenting paths if the residual network is recursively constructed by identifying the smallest length augmenting path from \(s\) to \(t\). 
\end{definition}
In our formulation, an edge represents a transition the system can make. For this reason, we assign unit capacities for every edge in the flow network. In unit-capacity networks, given a set \(S\) of augmenting paths characterizing maximum flow, any two augmenting paths \(P_1, P_2 \in S\) are always edge-disjoint with respect to each other. 
\begin{definition}[Minimum-cut edge~\cite{cormen2009introduction}]
A \emph{minimum-cut edge}, given a source \(s\) and sink \(t\), is an edge that belongs to some minimum cut. A minimum-cut edge can be identified by removing the edge and checking if the max-flow value decreases. For definition of minimum-cut, refer to Appendix or~\cite{cormen2009introduction}.
\end{definition}
\begin{definition}[\(ij\)-cycle]
Given a graph \(G\) and propositions characterizing the sequence specification, \(p_1, \ldots, p_n\), let \(P_i\) represent a simple path from \(v_i\) to \(v_{i+1}\) for all \(i = 1,\ldots, n\). The simple path from \(v_1\) to \(v_{n+1}\) characterized by (\(P_1, \ldots, P_n\)) has an \emph{\(ij\)-cycle} if there is an edge \((u,w)\in P_i\) and \((w,v)\in P_j\) for some \(i+1 \leq j \leq n+1\), except for the case in which both \(w = v_{i+1}\) and \(j = i+1\). 
\end{definition}
\begin{definition}[Sequence flow path]
Given a graph \(G\) and propositions characterizing the sequence specification, \(p_1, \ldots, p_n\), let \(P_i\) represent a simple path from \(v_i\) to \(v_{i+1}\) for all \(i = 1,\ldots, n\). The simple path from \(v_1\) to \(v_{n+1}\) characterized by \(P_f=\)(\(P_1, \ldots, P_n\)) is a sequence flow path if the subgraph of \(G\) induced by \(P_f\) satisfies equation~\eqref{eq:assumption2}. The notation \(P_f(i)\) is used to denote the \(i\)-th path \(P_i\) in \(P_f\). A sequence flow path has no \(ij\)-cycles.
\end{definition}
\begin{definition}[Sequence flow]
A sequence flow \(\mS_f = \{P_{f,1},\ldots, P_{f,\tilde{f}}\}\) for propositions characterizing the sequence specification is a finite set of sequence flow paths, \(P_{f,1},\ldots, P_{f,\tilde{f}}\), such that any two sequence flow paths are edge-disjoint with respect to each other, and any sequence flow path, \((P_{f,k_1}(1)\), \(P_{f,k_2}(2)\), \(\ldots\), \(P_{f,k_n}(n))\), where \(k_1,\ldots,k_n\leq \tilde{f}\), have no \(ij\)-cycles. 
\end{definition}
\subsection{Test Specification}
In this paper, we consider reachability specifications as mission objectives for the agent under test. For the test itself, we wish to observe a sequence-like behavior of the agent in its attempt to satisfy its mission objectives. Formally, this test behavior can be described by the following temporal logic formula,
\begin{equation}
    \varphi_{test} := \lozenge(p_1 \wedge \lozenge(p_2 \wedge \lozenge(\cdots \wedge \lozenge p_n))) \bigwedge_{i=1}^{n-1} (\neg p_{i+1} \mathsf{U}\, p_{i})\, ,
    \label{eq:test_fragment}
\end{equation}
where \(p_1, \ldots, p_n\) are propositional formulas. This is a sequence-like formula since the agent has to eventually visit every \(v_i\), but it cannot visit \(v_{i+1}\) before visiting \(v_i\), where \(v_i \vdash p_i\), for all \(i = 1, \ldots, n\).
The agent under test does not have access to the test specification \(\varphi_{test}\). Since LTL formulae cannot be evaluated on finite test runs, the length of the test run depends on the time the agent takes to satisfy its mission objective. 
\textit{Example: } Consider a gridworld in Figure~\ref{fig:illustration} on which the agent can transition between states (up, down, left, right) with the mission specification of reaching some goal state (formalized as \(\varphi_g = \lozenge g\)). Of the many possible paths the agent can take to meet its objective, we're interested in observing it navigate to the goal while restricted to a class of paths described by the test specification \(\varphi_{test} = \lozenge(p_1 \wedge \lozenge p_2)\). How would we constrain actions of the agent in certain states, such that it navigates through the sequence of waypoints before reaching the goal? Furthermore, is it possible to synthesize these constraints such that the sequence flow value from \(p_1\) to \(g\) is maximized?
\begin{figure}[h]
    \centering
    \begin{tabular}{ll}
    \includegraphics[width=.45\linewidth]{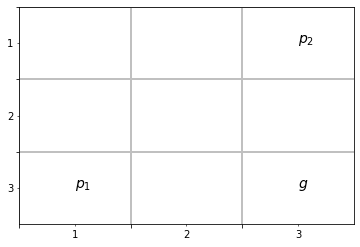} \,\,
    \includegraphics[width=.45\linewidth]{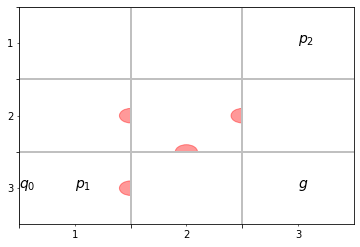}
    \end{tabular}
    \caption[]{{\small{Left: Unconstrained gridworld labeled by propositional formulas. Right: A test environment synthesized by our algorithm where the transitions \((2,1)\rightarrow (2,2)\), \((3,1)\rightarrow (3,2)\), \((2,2)\rightarrow (3,2)\), and \((2,2)\rightarrow (2,3)\) blocked. Red semi-circle patches illustrate one-way constraints, that is, transition from state \(u\) to state \(v\) is restricted, but \(v\) to \(u\) is allowed, if arch of the semi-circle is in the grid corresponding to \(u\) along the transition from \(u\) to \(v\).}}}
    \label{fig:illustration}
\end{figure}
\subsection{Problem Statement}
Now, we formalize the test environment synthesis problem. We limit our focus to \emph{static} test environments, by which we mean that the test environment does not react to the actions of the agent during the test, leaving the reactive test synthesis problem for future work.
\begin{definition}[Test Graph]\label{def:testGraph}
Given a labeled directed graph \(G = (V,E)\), a mission/agent specification \(\varphi_A = \lozenge p_{n+1}\), a test specification \(\varphi_{test}\) (equation~\ref{eq:test_fragment}),
a \emph{test graph} \(G' = (V, E\backslash C)\) is the directed graph obtained by removing set of edges, \(C\), from the original graph \(G\). On \(G'\), a run \(\sigma\) starting from state \(v_1\) will satisfy the specification,
\begin{equation}
    \sigma \models \lozenge p_{n+1} \iff \sigma \models \lozenge(p_1 \wedge \lozenge(p_2 \wedge \lozenge(\cdots \wedge \lozenge p_{n+1}))) \bigwedge_{i=1}^{n} (\neg p_{i+1} \mathsf{U}\, p_{i})\,.
    \label{eq:test_spec}
\end{equation}
\end{definition}
\begin{definition}[Minimally Constrained Test Graph]
A test graph \(G'\) is \emph{minimally constrained} if the sequence flow value from \(v_1 \models p_1\) to \(v_{n+1} \models p_{n+1}\) on \(G'\) are maximized.
\end{definition} 
\begin{problem}
Given an agent specification \(\psi_A = \lozenge v_{n+1}\), a labeled directed graph \(G = (V,E)\) induced by the non-deterministic transition model \(\mathcal{T}\) of the agent, a test specification \(\varphi_{test} = \lozenge (p_1 \wedge \lozenge (p_2 \wedge \lozenge(\cdots \wedge \lozenge p_n)))\), 
static constraints \(C\subseteq E\) such that on the resulting graph \(G' = (V, E \backslash C\)) is a minimally-constrained test graph.
\label{Pr:Problem1}
\end{problem}
Standard algorithms, such as Edmonds-Karp~\cite{cormen2009introduction}, can find the minimum-cut of a single source-sink flow problem in a graph \(G=(V,E)\) in \(O(|V||E|^2)\) time. However, here we aim to find a cut that maximizes the flow from a waypoint \(p_i\) to its consecutive waypoint \(p_{i+1}\), while eliminating any flow to waypoints \(p_j\) (\(j > i+1\)) for all \(i = 1,\cdots,n\). In other words, some flows need to be cut while other flows should be maximized. The problem of constructing a minimally constrained test graph for observing a sequence-like specification can be cast as the following optimization,
\begin{equation}
\begin{aligned}
\max_{\substack{C \subset E \\ G' = (V,E\backslash C)}} \quad & f_{G'}\\
\textrm{s.t.} ~~\quad\quad & f_{G'} \leq f_{G'}(v_i, v_{i+1})\, && \forall i = 1,\cdots,n,\\
\quad & f_{G'}(v_i, v_j) = 0\, && \forall i = 1,\cdots, j-2, \forall j = 3, \cdots, n,\\
\end{aligned}
\label{eq:ILP_Problem}
\end{equation}
where the variables are \(C\subset E\), the set of edges to be constrained, scalar \(f_{G'}\), scalars \(f_{G'}(v_i, v_j)\) that represent the total flow from source \(v_i\) to sink \(v_j\) on a unit-capacity graph \(G'\), and the problem data is the original graph \(G = (V,E)\) and the nodes \(v_1, \ldots, v_n\). Solving this optimization directly will require constructing an ILP, for which constructing the constraint set is not straightforward. Furthermore, it would require solving an ILP with \(|E|\) number of integer variables. As a result, we propose Algorithm~\ref{alg:StaticConstraints} to iteratively find  the optimal \(f_G'\).
\section{Algorithm for Synthesizing Test Environment}
\label{sec:Algorithms}
Let \(G = (V, E)\) be a directed graph, with unit capacity on every edge, induced by the Kripke structure \(\mT\) of the system under test. Assuming that the test environment has complete freedom to ``block" any transition in the graph \(G\), Algorithm~\ref{alg:StaticConstraints} returns a set of edges, \(C\subset E\), of the graph \(G\) that must be removed before the test run. Let \(d_{G}(v_1, v_2)\) denote the length of the shortest path from vertex \(v_1\) to vertex \(v_2\) on graph \(G\). First, we make following assumptions on \(G\),
\begin{assumption}
 Likewise, for each \(i \in \{1,\ldots, n+1\}\), let \(v_i\) denote the vertex \(v \in V\) s.t \(v\vdash p_i\). Assume \(|v_i| = 1\), for all \(i = \{1, \ldots, n+1\}\). 
\label{assump:n_goal_nodes}
\end{assumption}
Informally, assumption~\ref{assump:n_goal_nodes} states that every propositional formula, \(p_1, \ldots, p_{n+1}\), has a single vertex in \(G\) associated with it.
\begin{assumption}
 There exists a set of edges \(C \subseteq E\) such that the modified graph obtained by removing these edges, \(\mG' = (V, E\backslash C)\), is such that
\label{assump:exist_config}
\end{assumption}
 \begin{equation}
     d_{\mG'}(v_1, v_{n+1}) > \cdots > d_{\mG'}(v_n, v_{n+1}) > d_{\mG'}(v_{n+1}, v_{n+1}) = 0.
     \label{eq:assumption2}
 \end{equation}
The above Assumption~\ref{assump:exist_config} is equivalent to the statement that by removing some edges (or restricting certain transitions) from the original graph \(G\), there exists some set of initial conditions \(Q_0\) for which the only path(s) to the goal \(g\) is through the behavior \(\varphi_{test}\). This assumption is imperative since there might be instances for which it is impossible to construct a test graph. For example, in the following simple labeled graph (Figure~\ref{fig:invalid_config}), it is impossible to construct a test graph for the test specification \(\varphi_{test} = \lozenge(p_1 \wedge \lozenge p_2)\). Once the system is in state \(v_1\), it can directly proceed to the goal state \(v_g\) without visiting \(v_2\). For instances such as this one, a reactive test environment is necessary.
\begin{figure}[H]
    \centering
    \begin{tikzpicture}
          \Vertex[label=$v_2$]{A}
          \Text[x=0, y=0.5]{$p_2$}
          \Vertex[x=1.5, y=0, label=$v_1$]{B}
          \Text[x=1.5, y=0.5]{$p_1$}
          \Vertex[x=3, y=0, label=$v_g$]{C}
          \Text[x=3, y=0.5]{$g$}
          \Edge[bend=45, Direct=true](A)(B)
          \Edge[bend=45, Direct=true](B)(A)
          \Edge[Direct=true](B)(C)
    \end{tikzpicture}
    \caption[]{\small{An invalid configuration of propositional formulas for test specification \(\varphi_{test} = \lozenge (p_1 \wedge \lozenge p_2)\) }}
    \label{fig:invalid_config}
\end{figure}
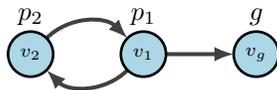
\subsection{Finding Combinations of Augmenting Paths}
In a unit-capacity flow network \(G = (V,E)\) with source \(s\) and sink \(t\) with maximum flow value \(f\), let \(S = \{P_1, \ldots, P_f\}\) represent the set of augmenting paths, or edge-disjoint paths, that characterize an instance of the maximum flow in the network since there can be multiple flow functions that can achieve maximum flow on the graph from source \(s\) to sink \(t\). For maximum flow from source \(s\) to sink \(t\) on \(G\), let maximum flow value be denoted \(f_G(s, t)\). On a minimally constrained test graph \(G'\), the maximum sequence flow value will be bounded as follows, 
\begin{equation}
    f_{G'}(v_1, v_{n+1}) \leq \min_{\substack{i=1,\ldots,n}} f_{G'}(v_i, v_{i+1}).
\end{equation}
Let \(F_i = \{S_{1,i}, \ldots, S_{k_i,i}\}\) denote the set of all sets of augmenting paths that characterize the maximum flow. Note that \(F_i\) is finite since the number of edges are finite, but can be combinatorial in number of simple paths from \(v_i\) to \(v_j\). Then, \(|F_i|\) denotes the total number of maximum flow realizations from \(v_i\) to \(v_{i+1}\). 
Consequently, the total number of augmenting path combinations from \(v_1\) to \(v_{n+1}\) will be \(\Pi_{i=1}^{n} |F_i|\). However, not every augmenting path combination might lead to a valid test graph since there could exist a combination of augmenting paths that violates the requirement in equation~\eqref{eq:test_spec} by constructing an \(ij\)-cycle. Consider the simple example of the \(3\times 3\) grid in Figure~\ref{fig:cycle_demo}. The combination of sequence flows \((S_{11}, S_{22}, S_{13})\) will give us \(f_{G'}(v_1, v_{n+1}) = 1\), but the combination of \((S_{11}, S_{12}, S_{13})\) forms an \(ij\)-cycle.
\begin{figure}[!htbp]
    \centering
    \begin{tabular}{ll}
    \includegraphics[width=.47\linewidth]{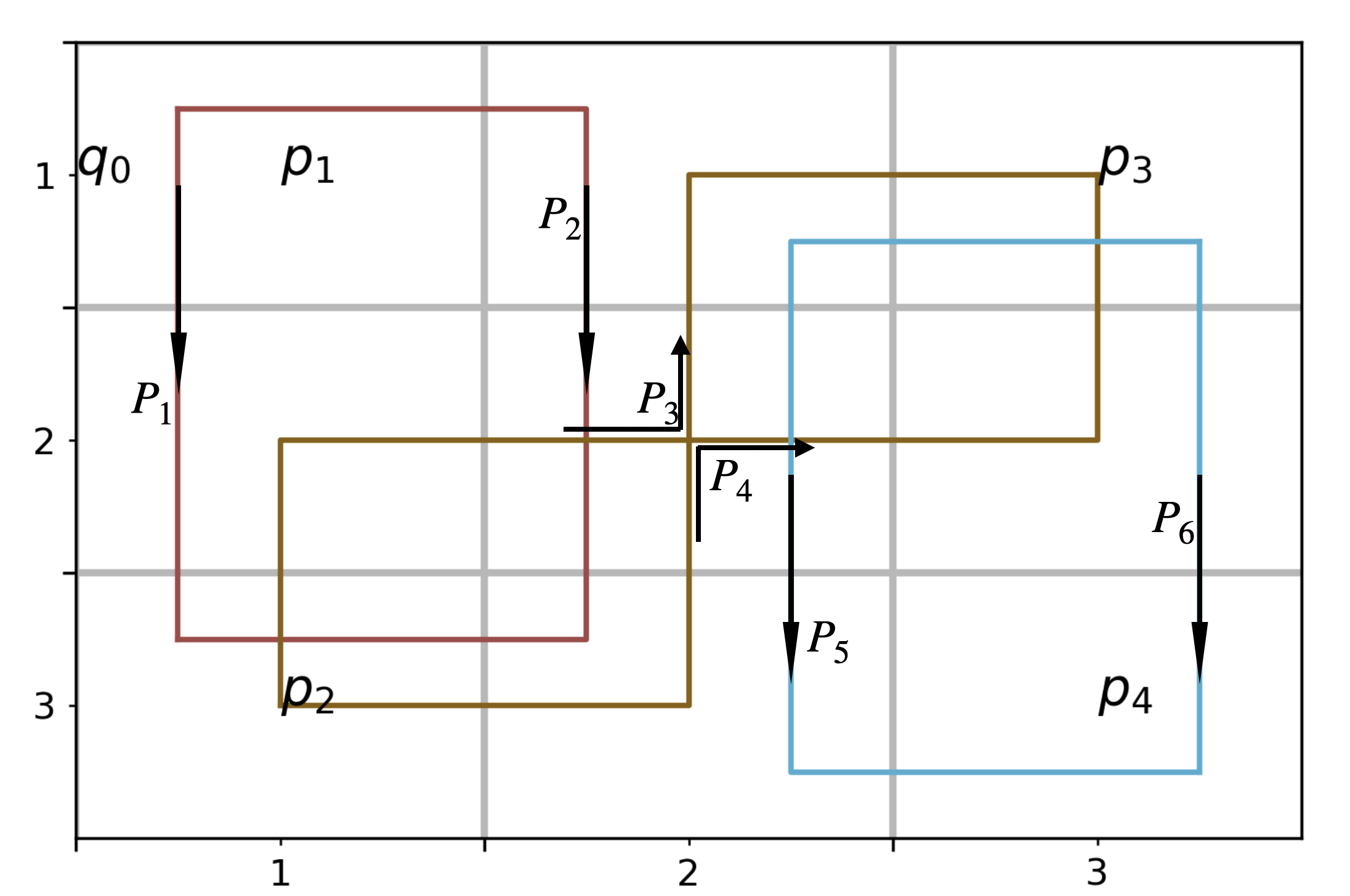} \,\,
    \includegraphics[width=.47\linewidth]{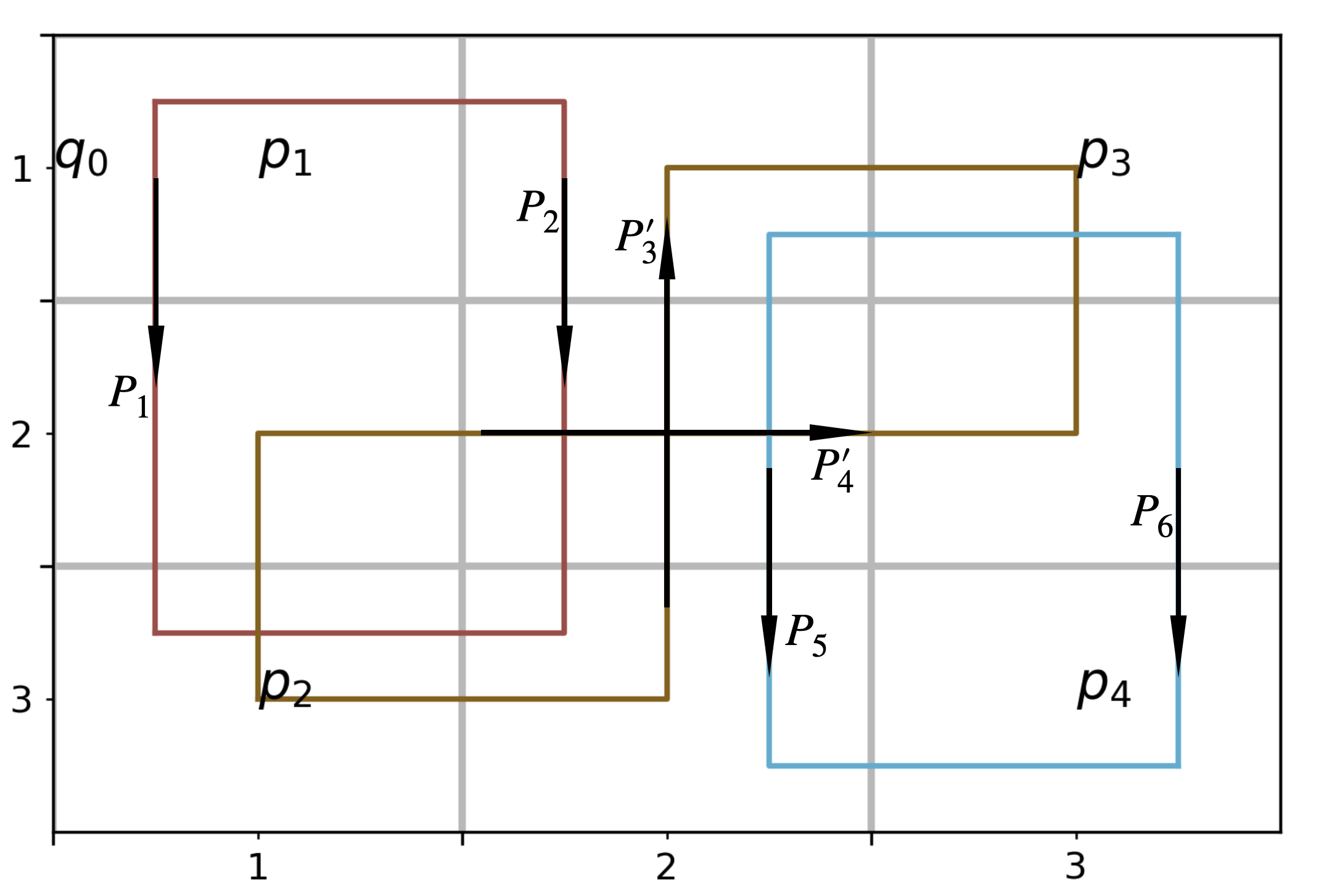}
    \end{tabular}
    \caption[]{\small{In this \(3\times 3\) grid, the left and right figures illustrate two different augmenting path combinations. In both figures, the augmenting paths characterizing the flow from \(p_1\) to \(p_2\) and \(p_3\) to \(p_4\) are the same: \(S_{11} = \{P_1, P_2\}\) characterizes the maximum flow from \(p_1\) to \(p_2\), and \(S_{13} = \{P_5, P_6\}\) characterizes the maximum flow from \(p_3\) to \(p_4\). On the left, \(S_{12} = \{P_3, P_4\}\), and on the right, \(S_{22} = \{P'_3, P'_4\}\). It is possible to form a sequence flow on the right with \((P_1, P'_3, P_6)\), but not on the left.}}
    \label{fig:cycle_demo}
\end{figure}
To avoid this issue, the algorithm searches through all combinations of sequence flows before constructing the input to the ILP~\eqref{eq:ILP_Problem}. Since this is an expensive computation, a further assumption on the input graph and set of propositions can ease this bottleneck. This assumption is as follows, 
\begin{assumption}
Let \(F^{min}_i = \{S^{min}_i = \{P_1,\ldots, P_{f_G(v_i, v_{i+1})}\}\}\), represent the set of sets of shortest augmenting paths that characterizes the flow from \(v_i\) to \(v_{i+1}\) on \(G\). Then, there exists a combination \((S^{min}_1, \ldots, S^{min}_n)\) on which a maximum sequence flow can be characterized.
\label{assump:SAPs}
\end{assumption}
In other words, Assumption 3 allows us to reason over combinations of shortest augmenting path flows, which is combinatorial in all shortest paths, instead of combinations of all augmenting flows, which is combinatorial in all simple paths. All shortest paths a subset of all simple paths between two nodes. 
\subsection{Iterative Synthesis of Constraints}
Now we present the ILP used to iteratively compute cuts in Algorithm~\ref{alg:StaticConstraints},
\begin{equation}
\begin{aligned}
 \max_{\substack{x \in \mathbb{B}^n,\, f\in \mathbb{B}^l \\ b\in \mathbb{B}^m}} \quad & \bbone^Tf\\
\textrm{s.t.} ~~\quad & A_{cut}x \geq \bbone\, && \\
\quad & A_{keep}x \leq D_{keep} b\, && \\ 
  \quad & b \leq A_{keep}x \, && \\
  \quad & D_ff \leq A_f(\bbone - b)\,, &\\
  \quad & f \geq A_f(\bbone -b) - D_f\bbone + \bbone, &\\
\end{aligned}
  \label{eq:ILP_Problem}
\end{equation}
where (\(x,b,f\)) are the optimization variables, and \(A_{cut} \in \mathbb{B}^{k\times n}\), \(A_{keep}\in \mathbb{B}^{m\times n}\), \(D_{keep}\in \mathbb{B}^{m\times m}\), \(D_f\in \mathbb{B}^{l\times l}\), \(A_f\in \mathbb{B}^{l\times m}\) are problem data described in more detail in Appendix~\ref{sec:appendix_constraints}.

\begin{algorithm}[h]
\caption{Restrict Transitions}
\begin{algorithmic}
\Require \(\varphi_{test}\), \(\varphi_a\), \(G = (V, E, L)\).
\Ensure \(C\subseteq E\).
\State \(p \leftarrow \{p_1, \ldots, p_n, p_{n+1}\}\), \(V_p \leftarrow \{v_1, \ldots, v_n, v_{n+1}\}\)  \Comment \(v_{i} \vdash p_i\)
\State \(P_{cut}\leftarrow \)Find-Cut-Paths(\(G, p\))
\State \(C = \{\}\)
\If{Assumption~\ref{assump:SAPs}}
\State \(flg \leftarrow 1\)
\EndIf
\While{\(P_{cut} \neq \emptyset\)} \Comment Repeat until all cuts are found
\State \(E \leftarrow \text{Edges in } P_{cut}\) 
\State \(\mA, \mP_{keep}, |\mA|, \tilde{f} \leftarrow\) Sequence-Flows\((G, p, flg=0)\) \Comment Combinations of sequence flows
\ForAll{\(j = 0,\ldots,|\mF|\)}
\State \(\mA_f \leftarrow \mA(j)\) \Comment Selecting a combination \((S_1,\ldots, S_n)\) 
\State \(P_{keep} \leftarrow \mP_{keep}(j)\) \Comment Augmenting paths for each \(v_i\) to \(v_{i+1}\)
\State \(MC_{keep} \leftarrow\) Min-Cut-Edges\((G, p, P_{keep})\) 
\State \(D_{keep}\leftarrow diag(A_{keep}\bbone)\)
\ForAll{\(A_f \in \mA_f\)}
\State \(D_f \leftarrow diag(A_f\bbone)\)
\State \(A_{cut}, A_{keep}, D_{keep}\leftarrow\) ILP-params\((P_{cut}, P_{keep}, MC_{keep})\)
\State \(x^*, f^*, b^*\leftarrow \) ILP\((A_{cut}, A_{keep}, D_{keep}, A_f, D_f)\)\Comment Call to ILP~\eqref{eq:ILP_Problem}
\If{\(\bbone^T f^* = \tilde{f}\)}
\State \(C_{new} \leftarrow \{e_i| x^*_i = 1\}\)
\State \(C \leftarrow C \cup C_{new}\)
\State \textbf{break} \Comment Breaking out of both for loops
\EndIf
\EndFor
\EndFor
\State \(G \leftarrow G\backslash C_{new}\)
\State \(P_{cut}\leftarrow \) Find-Cut-Paths(\(G, p\))
\EndWhile
\end{algorithmic}
\label{alg:StaticConstraints}
\end{algorithm}
\par 
The subroutine Find-Cut-Paths, described in the algorithm takes as input a graph \(G\) and list of propositions, \(p\), and uses max-flow algorithms to find a set of augmenting paths for every source-sink pair (\(v_i, v_j\)) on \(G_{ij} = G\backslash (V_p\backslash\{v_j, v_i\})\) such that \(i+1\leq j \leq n+1\). All of these augmenting paths are collectively returned as the output \(P_{cut}\), and the edges constituting these cuts are denoted by \(E_{cut}\). Note that \(P_{cut}\) does not return all simple paths from \(v_i\) to \(v_{j>i+1}\), but just a set of edge-disjoint paths. As a result, we need to iteratively restrict transitions until \(P_{cut}\) is empty. A summary of the Minimum-Cut-Edges and the Sequence-Flow subroutine is given in the Appendix~\ref{sec:appendix}.
Finally, the parameters to ILP~\eqref{eq:ILP_Problem} are constructed as follows:
\medskip
\label{sec:ILP_description}\\
\textit{1. Parameters:} The parameters used to construct the problem data for the ILP~\eqref{eq:ILP_Problem} are the set of paths that need to be constrained, \(P_{cut}\), the set of paths whose combination constitutes sequence flow and should not be constrained, \(P_{keep}\), and the set of minimum-cut edges, \(MC_{keep}\), on the paths constituting \(P_{keep}\). The set \(\mP_{keep} = \{(S_{m_i,1},\ldots,S_{m_n,n})|S_{m_i, i} \in F_i, \, 0\leq m_i \leq k_i\}\) is a set of all augmenting path combinations. For a given combination of sets of augmenting paths, \(P_{keep}=(S_{m_1,1},\ldots,S_{m_n,n})\), with the cardinality of \(S_{m_i,i}\) being denoted as follows, \(n_i := |S_{m_i,i}|\), and \(m:= \Sigma_{i=1}^n\). Suppose a combination of augmenting paths, \(S_f = \{P = (P_1,\ldots, P_n)| P_i \in S_{m_i,i}\}\), represents a sequence flow, then a matrix \(A_f \in \mathbb{B}^{|S_f|\times m}\) can be constructed to represent the sequence flow \(S_f\). This construction is outlined in the descriptions of Constraints of the ILP. An instance of \(P_{keep}\) can have several sequence flows, \(S_f\), and correspondingly, several matrices, \(A_f\), all of which are collectively denoted by \(\mA_f\). The set of all such \(\mA_f\) is denoted by \(\mA\), which has cardinality \(|\mA| = |\mP_{keep}|\), since each \(\mA_f\) corresponds to an instance of \(P_{keep}\). The maximum sequence flow value is given by \(\tilde{f}\).\medskip
\\
\textit{2. Variables:}
The variable \(x\in \mathbb{B}^n\), where \(n=|E_{cut}|\), is the Boolean vector corresponding to edges \(E_{cut}\) such that for some \(k\leq n\), if \(x_k = 1\), then the corresponding edge is constrained, and \(x_k = 0\) means that it is left in the graph for future iterations. 
Given \(P_{keep} = (S_{m_1,1},\ldots,S_{m_n,n}) \in \mP_{keep}\), a combination of set of augmenting paths, the variable \(b \in \mathbb{B}^m\) keeps track of whether an augmenting path in some \(S_{m_i,i}\) (\(1\leq i \leq n\)) is constrained or not. For some \(k\leq m\), if \(b_k = 1\), then the corresponding augmenting path in some \(S_{m_i,i}\) has minimum-cut edge(s) constrained by the ILP, and \(b_k = 0\) if none of the minimum-cut edges of that augmented path have been constrained. The variable, \(f \in \mathbb{B}^l\), represents the sequence flow vector for a given sequence flow, \(S_f\), such that \(l = |S_f|\) is the number of edge-disjoint paths constituting the sequence flow.\medskip \\
\textit{3. Constraints: } The constraints of ILP~\eqref{eq:ILP_Problem} are outlined here but futher details are given in Appendix~\ref{sec:appendix_constraints}.
The first constraint of the ILP, \(A_{cut}x\geq \bbone\), enforces the requirement that each path in \(P\in  P_{cut}\) is constrained. 
In the second and third constraints, \(A_{keep}x \leq D_{keep}b\) and \(b \leq A_{keep}x\), is used to determine the variable \(b\) from the variable \(x\). 
These two constraints ensure that for some \(q \leq n\), \(b_q = 1\) iff at least one minimum-cut edge on the path corresponding to the \(q\)-th row of \(A_{keep}\) is constrained, and \(b_q = 0\) iff none of the minimum-cut edges on the path corresponding to the \(q\)-th row of \(A_{keep}\) are constrained.
\\
The fourth and fifth constraints, \(D_ff \leq A_f(\bbone - b)\) and \(f \geq A_f(\bbone -b) - D_f\bbone + \bbone\), determine the flow value for a given set of sequence flow paths, \(S_f\). 
The fourth constraint ensures that if any of the constituent paths, \(P_1, \ldots, P_n\), in the \(q\)-th sequence flow path \(P = (P_1, \ldots, P_n)\in S_f\) (for \(1\leq q\leq l\)), is constrained, then the flow value, \(f_q = 0\). The last constraint ensures that if none of the constituent paths, \(P_1, \ldots, P_n\), in the \(q\)-th sequence flow path \(P = (P_1, \ldots, P_n)\in S_f\) (for \(1\leq q\leq l\)), are constrained, then the flow value, \(f_q = 1\). \medskip\\
\textit{4. Cost Function: }The cost function computes the maximum sequence flow value.\medskip\\
Algorithm~\ref{alg:StaticConstraints} does not proceed to the next iteration of \(P_{cut}\) until it finds the set of static constraints that return the maximum possible sequence flow value, \(\tilde{f}\). To guarantee completeness of Algorithm~\ref{alg:StaticConstraints}, we need to prove that the cuts synthesized in prior iterations do not preclude feasibility of further iterations with regards to assumption~\ref{assump:exist_config}. See Appendix~\ref{sec:appendix_complexity} for complexity of the subroutines in Algorithm~\ref{alg:StaticConstraints}.
\begin{lemma}
In a graph \(G = (V,E)\), let \(\mP\) represent a maximal set of sequence flow paths from \(v_1\) to \(v_n\). Let \(\mP_{cut}\) be the set of paths that need to constrained, with the edges constituting the paths in \(\mP_{cut}\) denoted by \(E_{cut}\subset E\). Then, the set of constraint edges \(C\subseteq E_{cut}\) can be found such that \(C\) does not constrain any path in \(\mP\).
\label{lemma:min_cut_edges}
\end{lemma}
\begin{proof}
See Appendix~\ref{sec:appendix_proofs}
\end{proof}
\begin{proposition}
Let \(G_{m} = (V, E_{m})\) denote the graph for which the \(m\)-th iteration of the ILP~(\ref{eq:ILP_Problem}) synthesizes new cuts \(C_m \subset E_m\). Then, Assumption~\ref{assump:exist_config} is satisfied on \(G_{m+1} = (V, E_m\backslash C_m)\).
\label{prop:Invariant}
\end{proposition}
\begin{proof}
In the first iteration, from Assumption~\ref{assump:exist_config}, we know there exists at least one test graph \(G' = (V, E\backslash C)\) that satisfies equation~\eqref{eq:assumption2}. Assume that the \(m\)-th iteration graph \(G_m = (V,E_m)\) also satisfies Assumption~\ref{assump:exist_config}. We will show by induction that the graph resulting from the the \((m+1)\)-th iteration, \(G_{m+1} = (V, E_m\backslash C_m)\), also satisfies Assumption~\ref{assump:exist_config}. By construction, Algorithm~\ref{alg:StaticConstraints} chooses a combination of set of augmenting paths \((S_1, \ldots, S_n)\) such that there exists a non-empty set of sequence flow paths \(\mF = \{(P_1, \ldots, P_n)|\, P_i \in S_i\}\) such that the simple path from \(v_1\) to \(v_n\) characterized by \(\Gamma = (P_1, \ldots, P_n) \in \mF\) does not form an \(ij\)-cycle for some \(i<j\leq n\). This implies that on the subgraph comprising of the edges in \(\Gamma\), equation~\eqref{eq:assumption2} is satisfied. \par
If the maximum possible sequence flow in a minimally constrained test graph is \(f^*\), then we can find a combination \((S_1, \ldots, S_n)\) such that for each \(i = 1,\ldots,n\), there exists a set \(S_i^'=\{P_{1,i}, \ldots, P_{f^*,i}\} \subseteq S_i\), from which we can construct the set \(\mF' = \{(P_{k_1,1},\ldots, P_{k_n, n})|P_{k_i, i}\in S'_i, \, 0\leq k_i \leq f^*\}\subseteq \mF\). By construction of the input variables to the ILP~\eqref{eq:ILP_Problem}, the constraints of ILP~\eqref{eq:ILP_Problem} require that the sequence flow variable \(f\) has atleast one element that is 1. This is possible only if there exists a set of edges \(C_m\) that constrain \(A_{cut,m}\) such that there exists at least one sequence path \(P\in\mF\) that does not have any of its minimum-cut edges constrained, which is true as shown in Lemma~\eqref{lemma:min_cut_edges}. Therefore, the new graph \(G_{m+1} = (V,E\backslash C_m)\) satisfies Assumption~\eqref{assump:exist_config}. \qed
\end{proof}
\begin{theorem}
Under Assumption~\eqref{assump:exist_config}, Algorithm~\ref{alg:StaticConstraints} is complete and returns a test graph \(G'\) from Definition~\ref{def:testGraph} that satisfies equation~\eqref{eq:assumption2}.
\label{thm:valid_test_spec}
\end{theorem}
\begin{proof}
Consider iteration \(m\) of the outer while loop in Algorithm~\ref{alg:StaticConstraints}, and let the graph at the \(m\)-th iteration be \(G = (V, E_m\). Denote \(V_p = \{v_i|v_i\vdash p_i,\, \forall 1\leq i\leq n+1\}\). Let \((n_{i,j})_m\) denote the maximum flow value from \(v_{i}\) to \(v_{j}\) on \(G_{ij} = (V\backslash(V_p\backslash \{v_i, v_{j}\}), E_m)\), for some \(i,j\) such that \(1\leq i<j-1 \leq n\). This implies that there is a set \(S_{i,j}\) of \((n_{i,j})_m\) edge-disjoint paths that characterize the maximum flow from \(v_i\) to \(v_j\) on \(G_{ij}\). Let \((\mP_{i,j}(k))_m\) be the set of all simple paths from \(v_i\) to \(v_j\) that share an edge with the \(k\)-th path in \(S_{i,j}\). Let \((MC_{i,j})_m\) be the set of minimum-cut edges on the paths in \(S_{i,j}\) and let \((E_{i,j})_m \subset E_m\) be the set of all edges on some path from \(v_i\) to \(v_{j+1}\) on \(G_{ij}\). Clearly, \((MC_{i,j})_m \subseteq (E_{i,j})_m\). \par
For every \(m\geq 1\), we can claim that \(|(E_{i,j})_{m+1}| < |(E_{i,j})_m|\) because edges are removed to constrain \(S_{i,j}\) in the \(m\)-th iteration. Let \(\tilde{m}\) be the number of iterations for \(G_{ij}\) to become disjoint. In the worst-case, edges continue to be constrained until iteration \(\tilde{m}\) at which \((E_{i,j})_{\tilde{m}} = (MC_{i,j})_{\tilde{m}}\), at which point constraining edges to cut \((S_{i,j})_{\tilde{m}}\) results in a cut separating \(v_i\) and \(v_j\). Thus, \(\tilde{m}\) has to be finite for every such \(i,j\).\par
At the same time, from Proposition~\ref{prop:Invariant}, the synthesized cuts are such that Assumption~\ref{assump:exist_config} is maintained as an invariant. Therefore, when the last set of paths \(S_{i,j}\) are constrained, the final test graph \(G'\) is such that \(d_{G'}(v_1,v_{n+1}) > \ldots >d_{G'}(v_n,v_{n+1})\).
\end{proof}
In addition to Assumption~\ref{assump:exist_config}, if Assumption~\ref{assump:SAPs} holds, Algorithm~\ref{alg:StaticConstraints} can be modified by a parameter setting. The proof of Theorem~\ref{thm:valid_test_spec} still holds.
\begin{lemma}
On the test graph \(G'\), any test run \(\sigma\) starting from state \(v_1\) will satisfy the specification~(\ref{eq:test_spec}).
\label{lemma:test_graph_test_spec}
\end{lemma}
\begin{proof}
See Appendix~\ref{sec:appendix_proofs}
\end{proof}
From Theorem~\ref{thm:valid_test_spec} and Lemma~\ref{lemma:test_graph_test_spec}, Algorithm~\ref{alg:StaticConstraints} synthesizes a test graph \(G'\) for the test specification ~\eqref{eq:test_spec}, solving Problem~\ref{Pr:Problem1}.
\begin{proposition}
Consider the test graph $G'$ from Definition~\ref{def:testGraph} for the test specification $\sigma$ from~\eqref{eq:test_spec}. If Assumption~\ref{assump:exist_config} holds, Algorithm~\ref{alg:StaticConstraints} returns a minimally constrained test graph.
\label{prop:min_constr}
\end{proposition}
\begin{proof}
See Appendix~\ref{sec:appendix_proofs}.
\end{proof}
\section{Examples}
\label{sec: Examples}
We illustrate the iterative synthesis of restrictions on a simple graph and a small gridworld, and then show runtimes of Algorithm~\ref{alg:StaticConstraints} on random gridworld instances for both the case for which Assumption~\ref{assump:exist_config} is true, and the case for which Assumptions~\ref{assump:exist_config} and ~\ref{assump:SAPs} are true. \medskip\\
\textbf{Simple graph:}
Consider a simple non-deterministic Kripke structure representing an autonomous agent, shown in Figure~\ref{fig:simple_automaton}, with propositional formulas labeled adjoining the states. The agent mission objective is to reach \(g\) while being restricted to start from state \(q_0\). The test environment seeks to restrict transitions such that the agent is prompted to pass through waypoint \(w\) in its trajectory to \(g\).\par
Inputs to Algorithm~\ref{alg:StaticConstraints} include the labeled graph \(G\) induced by the Kripke structure, the agent specification \(\lozenge p_3\), the test specification \(\lozenge p_2\), and the initial condition constraint \(\lozenge p_1\). 
Algorithm~\ref{alg:StaticConstraints} constrains the edges \(\{(v_2,v_4), (v_4, v_6)\}\) in the first iteration, and the edges \(\{(v_2,v_5), (v_5, v_6)\}\) in the second iteration. Although in this simple example, searching the set of all augmented paths becomes searching over all paths, in larger examples discussed below, each augmented path represents a class of paths that share some edge(s) with it.\\\medskip
\begin{figure}[H]
    \centering
    \begin{tabular}{cc}
     \begin{tikzpicture}
          \Vertex[label=$q_0$]{A}
          \Text[x=0, y=0.3, position=above, distance=1mm]{$p_1$}
          \Vertex[x=1, y=0, label=$v_2$]{B}
          \Vertex[x=2, y=1, label=$w$, RGB, color={190,174,212}]{C}
          \Text[x=2.3, y=1.2, position=right, distance=1mm]{$p_2$}
          \Vertex[x=2, y=0, label=$v_4$]{D}
          \Vertex[x=2, y=-1, label=$v_5$]{E}
          \Vertex[x=3, y=0, label=$v_6$]{G}
          \Vertex[x=4, y=0, label=$g$]{H}
          \Text[x=4.3, y=0, position=right, distance=1mm]{$p_3$}
          \Edge[Direct=true](A)(B)
          \Edge[Direct=true](B)(C)
          \Edge[Direct=true](B)(D)
          \Edge[Direct=true](B)(E)
          \Edge[Direct=true](C)(G)
          \Edge[Direct=true](D)(G)
          \Edge[Direct=true](E)(G)
          \Edge[Direct=true](G)(H)
    \end{tikzpicture}
    &
    \begin{tikzpicture}
          \Vertex[label=$q_0$]{A}
          \Text[x=0, y=0.3, position=above, distance=1mm]{$p_1$}
          \Vertex[x=1, y=0, label=$v_2$]{B}
          \Vertex[x=2, y=1, label=$w$, RGB, color={190,174,212}]{C}
          \Text[x=2.3, y=1.2, position=right, distance=1mm]{$p_2$}
          \Vertex[x=2, y=0, label=$v_4$]{D}
          \Vertex[x=2, y=-1, label=$v_5$]{E}
          \Vertex[x=3, y=0, label=$v_6$]{G}
          \Vertex[x=4, y=0, label=$g$]{H}
          \Text[x=4.3, y=0, position=right, distance=1mm]{$p_3$}
          \Edge[Direct=true](A)(B)
          \Edge[Direct=true](B)(C)
          \Edge[Direct=true,style={dashed}, color=red](B)(D)
          \Edge[Direct=true,style={dashed}, color=red](B)(E)
          \Edge[Direct=true](C)(G)
          \Edge[Direct=true, style={dashed}, color=red](D)(G)
          \Edge[Direct=true, style={dashed}, color=red](E)(G)
          \Edge[Direct=true](G)(H)
    \end{tikzpicture}
    \end{tabular}
    \caption[]{\small{\emph{Left:} Simple Kripke structure representing states that the agent can occupy. The waypoint, \(w\), is highlighted in purple to indicate that transitions are restricted corresponding to propositional formula \(p_2 = L(w)\). \emph{Right:} A test graph. Dashed edges in red illustrate transitions that have been restricted/removed from the Kripke structure above.}}
    \label{fig:simple_automaton}
\end{figure}
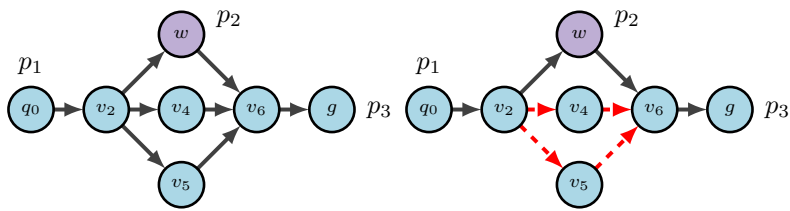
\textbf{Simple Gridworld: }
In Figure~\ref{fig:illustration_add_obstacles}, we illustrate the iterative synthesis of obstacles in a gridworld instance. Note that this configuration can be synthesized only by considering all sets of augmenting paths between (\(p_1, p_2\)) and (\(p_2, p_3\)). Since there is no shortest augmenting path from \(p_2\) to \(p_3\) that does not form a cycle with some (in this example, there is only one) shortest augmenting path flow from \(p_1\) to \(p_2\), it is imperative to use all sets of augmenting paths in the Sequence-Flows subroutine.

\begin{figure*}[h]
        \centering
        \begin{subfigure}[b]{0.475\textwidth}
            \centering
            \includegraphics[width=\textwidth]{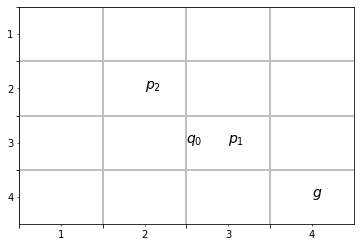}
            \caption[Network2]%
            {{\small Initial grid }}    
            \label{fig:(a)}
        \end{subfigure}
        \hfill
        \begin{subfigure}[b]{0.475\textwidth}  
            \centering 
            \includegraphics[width=\textwidth]{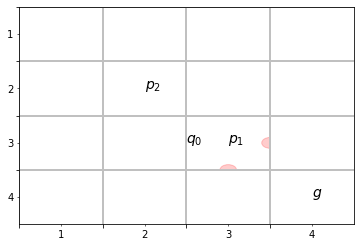}
            \caption[]%
            {{\small Iteration 1}}    
            \label{fig:(b)}
        \end{subfigure}
        \vskip\baselineskip
        \begin{subfigure}[b]{0.475\textwidth}   
            \centering 
            \includegraphics[width=\textwidth]{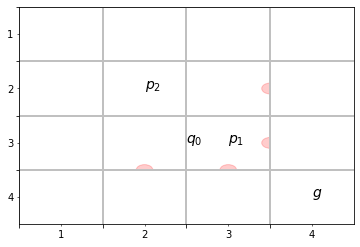}
            \caption[]%
            {{\small Iteration 2}}    
            \label{fig:(c)}
        \end{subfigure}
        \hfill
        \begin{subfigure}[b]{0.475\textwidth}   
            \centering 
            \includegraphics[width=\textwidth]{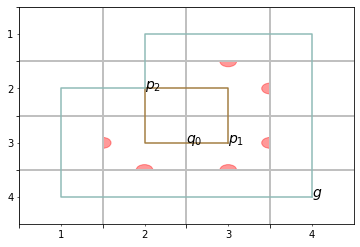}
            \caption[]%
            {{\small Iteration 3. The colored paths highlight a sequence flow from \(p_1\) to \(g\).}}    
            \label{fig:(d)}
        \end{subfigure}
        \caption[]
        {\small Synthesizing static test environment for \(\varphi_{test} = \lozenge(p_1 \wedge \lozenge p_2) \wedge \neg p_2 \mathsf{U}\, p_1\) and \(\varphi_{a} = \lozenge g\). } 
        \label{fig:mean and std of nets}
    \label{fig:illustration_add_obstacles}
\end{figure*}
\textbf{Random Gridworld Instances: }
For the case of setting all augmenting paths in the Sequence-Flows subroutine, we ran 50 random instances each for small gridworlds and propositions and plotted the average runtimes in Figure~\ref{fig:all_runtimes}. The number of propositions are limited by the size of the gridworld instances, which is restricted by the combinatorial nature of finding all sets of augmenting paths, and all combinations of sets of augmenting paths. 
\par
If we choose initial gridworld instances that satisfy Assumption~\ref{assump:SAPs}, then Algorithm~\ref{alg:StaticConstraints} can synthesize static constraints for slightly larger \(t\times t\) grid sizes. The average runtimes for 50 random iterations for various grid sizes \(t\) is plotted in Figure~\ref{fig:sap_runtimes}. The small increase to larger grid size is due to the Sequence-Flows subroutine reasoning over shortest augmenting paths, and not all augmenting paths.
\par
The average runtimes increase exponentially with the size of the grid. The number of propositions, denoted by \(|P|\), is labeled \(n\) if the test specification \(\varphi_{test}\)~\eqref{eq:test_spec} is comprised of propositions (\(p_1, \ldots, p_n\)). In both Figures~\ref{fig:all_runtimes} and ~\ref{fig:sap_runtimes}, the average runtime for fewer  propositions is at times higher that the average runtime for more propositions. This can be attributed to the Sequence-Flows subroutine taking longer to enumerate all simple paths (or all shortest paths in case of Assumption~\ref{assump:SAPs}) between two nodes, which could be greater in number due to fewer propositions constraining the graph.
\begin{figure}[!htbp]
\centering
    \begin{subfigure}[b]{0.475\textwidth}   
            \centering 
            \includegraphics[scale=0.3]{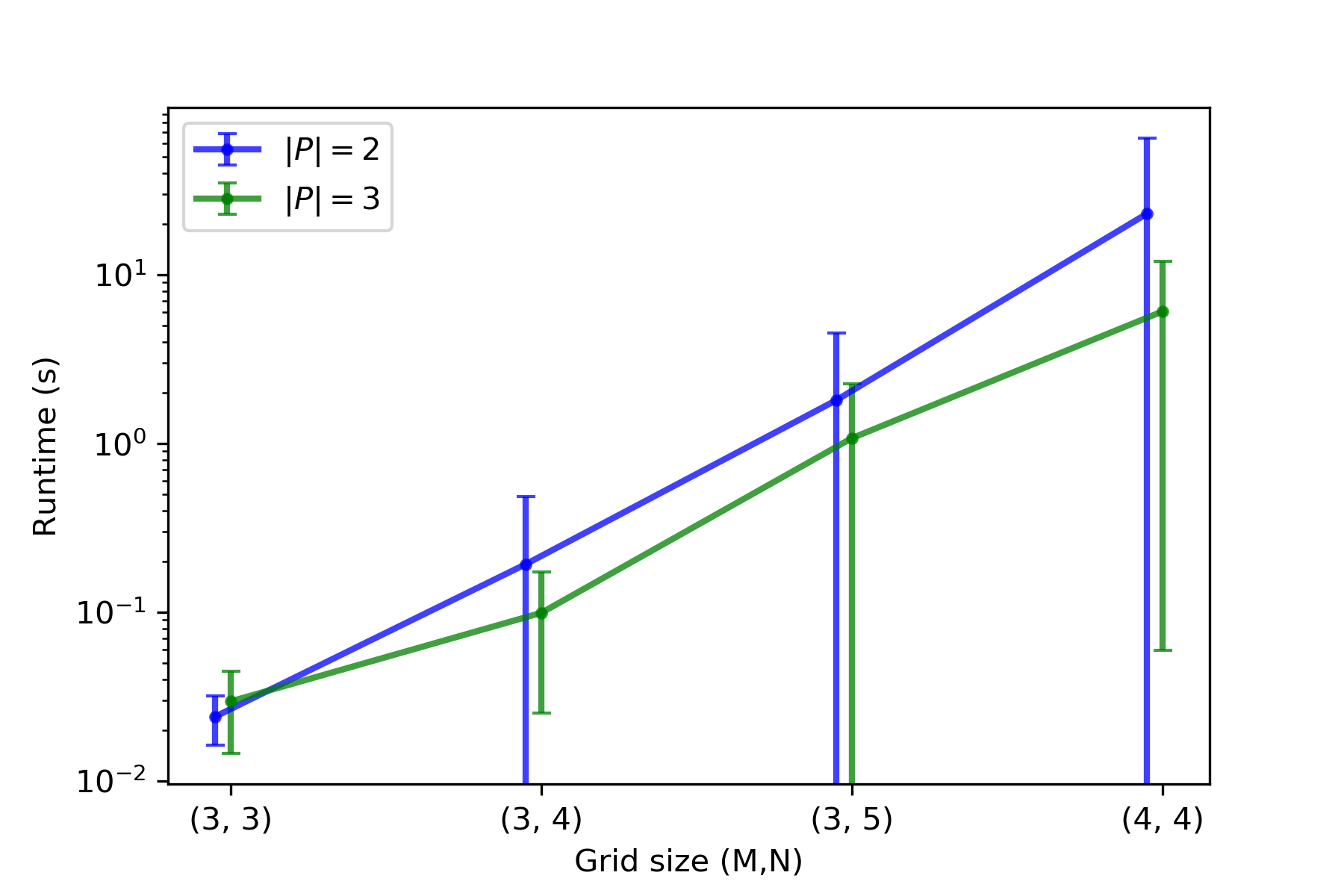}
            \caption[]%
            {{\small Small gridworld configurations using all augmenting flows.}}    
            \label{fig:all_runtimes}
        \end{subfigure}
        \hfill
        \begin{subfigure}[b]{0.475\textwidth}   
            \centering 
            \includegraphics[scale=0.4]{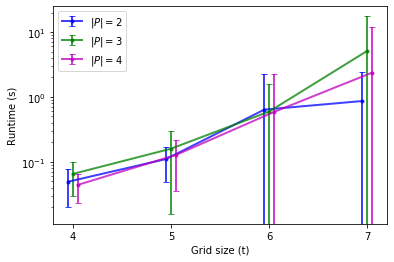}
            \caption[]%
            {{\small Gridworld configurations using only shortest augmenting flows.}}    
            \label{fig:sap_runtimes}
        \end{subfigure}
        \caption[]%
            {{\small Average runtime over 50 random instances. The number of propositions in \(\varphi_{test}\) is denoted by \(|P|\) in the legend. Error bars represent standard deviation of runtimes.}} 
\end{figure}

Another paradigm for the problem of synthesizing static test environments for sequence behaviors could be multi-commodity network flows, however, that setting considers multiple source-sink flows simultaneously drawing from the capacity of each edge, and here we compute separate network flows for every source-sink pair of nodes. 
\section{Conclusions and Future Work}
In this paper, we proposed an algorithm to synthesize a static test environment to observe sequence-like behavior in a discrete-transition system. First, we formulated this test environment synthesis problem as a problem of synthesizing cuts on graphs using concepts of flow networks. Then, we proposed an algorithm which synthesized the cuts iteratively using an integer linear program. We proved that this algorithm is complete, and that the edges constrained by the ILP at each iteration maintain feasibility of the constraint in the next iteration. Finally, we conducted numerical experiments on random gridworld instances to assess the runtime of our algorithm. Simulation results preclude this algorithm from being tractable to larger examples. As future work, we will investigate heuristic methods to solve this problem, and explore static and reactive test environment synthesis for other classes of test specifications.
\label{sec:Conclusions}

\subsection*{Acknowledgments}
The authors would like to acknowledge Professor Mani Chandy, Dr. Shih-Hao Tseng, and Dr. Ugo Rosolia for insightful discussions. We acknowledge funding from AFOSR Test and Evaluation Program, grant FA9550-19-1-0302. 

\bibliographystyle{plain}
\bibliography{bib}
\section{Appendix}
\label{sec:appendix}
\subsection{Definitions}
\begin{definition}
 A \emph{Linear Temporal Logic (LTL)} specification is inductively defined over a set of atomic propositions \(AP\) by the following syntax:
\begin{equation*}
    \varphi ::= p \, |\, \neg \varphi\, |\, \varphi_1 \wedge \varphi_2 \, |\, \bigcirc \varphi\, |\, \varphi_1 \mathsf{U} \varphi_2\, , 
\end{equation*}
\end{definition}
where \(p \in AP\) is an atomic proposition, \(\wedge\) (``and") and \(\neg\) (``not") are logical operators, and \(\bigcirc\) (``next") and \(\mathsf{U}\) (``until") are temporal operators. Other logical connectives such as \(\vee\) (``or") and \(\Rightarrow\) can be defined: 1) \(\varphi_1 \vee \varphi_2\)\(:= \neg(\neg\varphi_1 \wedge \neg \varphi_2)\), and 2) \(\varphi_1 \Rightarrow \varphi_2\) \(:= \neg \varphi_1 \vee \varphi_2\). Other temporal operators such as \(\lozenge\) (``eventually") and \(\square\) (``always") can also be defined: 1) \(\lozenge \varphi\) \(:= \text{true} \mathsf{U}\, \varphi\), and 2) \(\square \varphi\) \(:= \neg \lozenge \neg \varphi\). A \emph{propositional formula} \(p\) comprises only of logical connectives over atomic propositions. If a propositional formula, \(p\), is \emph{true} at state \(q\in Q\), we denote it as \(q \vdash p\).

\textit{Semantics: } An LTL formula \(\varphi\) is evaluated over a run \(\sigma = q_0 q_1 \ldots\) of the system. We denote \(\varphi\) holds at position \(j \geq 0\) of the run \(\sigma = q_0 q_1 \ldots\) by \(q_j \models \varphi\). For a given set of atomic propositions \(AP\) and a run \(\sigma\), the satisfaction of an LTL formula \(\varphi\) is defined inductively: 1) for an atomic proposition \(p\), \(q_j \models p\) if \(q_j \vdash p\), 2) \(q_j \models \neg \varphi\) if it is not true that \(q_j \models \varphi\), 3) \(q_j \models \varphi_1 \wedge \varphi_2\) if both \(q_j \models \varphi_1\) and \(q_j \models \varphi_2\), 4) \(q_j \models \bigcirc \varphi\) if \(q_{j+1} \models \varphi\), and 5) \(q_j \models \varphi_1 \mathsf{U} \varphi_2\) if \(\exists k\geq j\) such that \(q_k \models \varphi_2\) and \(\forall i, j\leq i < k\), \(q_i \models \varphi_1\). Similarly, \(q_j \models \lozenge \varphi\) if \(\exists k\geq j\) such that \(q_k \models \varphi\), and \(q_j \models \square \varphi\) if \(\forall k\geq j\), \(q_k \models \varphi\). A run \(\sigma\) satisfies \(\varphi\), denoted by \(\sigma \models \varphi\) if \(q_0 \models \varphi\).
\begin{definition}[Minimum cut]
 Given a flow network \(G = (V,E)\), a source node \(s\in V\) and a sink node \(t\in V\), a cut \((S,T)\) of \(G\) is such that \(s\in S\) and \(t\in T\), and \(S\cap T = \emptyset\) and \(S\cup T = V\). Cut edges are edges \((u,v)\in E\) such that \(u\in S\) and \(v\in T\), and denote the sum of capacities of the cut edges as the cut capacity: \(c(S,T) = \Sigma_{u\in S}\Sigma_{v\in T}c(u,v).\)
\end{definition}
\begin{definition}[Residual Network~\cite{cormen2009introduction}]
 Given a flow network \(G = (V,E)\) and a flow function \(f\) and capacity \(c\), a residual network \(G_f\) is one that has residual capacity function \(c_f\) defined as follows,
\begin{equation}
c_f(u,v) =
\begin{cases}
c(u,v) - f(u,v) & \text{if } (u,v) \in E,\\
f(u,v) & if (v,u) \in E, \\
0 & \text{otherwise.}
\end{cases}
\end{equation}
\end{definition}
\subsection{Descriptions of Subroutines of Algorithm~\ref{alg:StaticConstraints}}
\label{sec:appendix_subroutine}
The Min-Cut-Edges subroutine takes as input a graph \(G\), a list of propositions \(p\), and a non-empty set of augmenting paths for every source-sink pair (\(v_i, v_{i+1}\)) such that \(1\leq i\leq n+1\), and returns as output the set of minimum-cut edges on those augmenting paths. This information is then used in constructing the problem data for the ILP. The Sequence-Flows subroutine takes as input a graph \(G\), a list of propositions \(p\), and a parameter to indicate if Assumption~\ref{assump:SAPs} holds, computes the combination of all augmenting flows (or all shortest augmenting flows) that can result in a non-zero sequence flow from \(v_1\) to \(v_{n+1}\). It returns as output the set of all sets of matrices that capture sequence-flow paths, \(\mA\), a set of \(\mP_{keep} = \{(S_{m_i,1},\ldots,S_{m_n,n})|S_{m_i, i} \in F_i, \, 0\leq m_i \leq k_i\}\), the total number of combinations, \(|\mA|\), and the maximum possible sequence flow value, \(\tilde{f}\), which is determined when \(\mA\) is constructed.
\subsection{Constraints of the Integer Linear Program~\eqref{eq:ILP_Problem}}
\label{sec:appendix_constraints}
The first constraint of the ILP, \(A_{cut}x\geq \bbone\), enforces the requirement that each path in \(P\in  P_{cut}\) is constrained. Each row of \(A_{cut}\) corresponds to a path \(P \in P_{cut}\). The \(q\)-th row of \(A_{cut}\) is constructed as follows,
    \begin{equation}
        (A_{cut})_{q,r} = \begin{cases}
      1 & \text{if}\ E_{cut}(r) \in P = P_{cut}(q) \\
      0 & \text{otherwise.}
    \end{cases}
        \label{eq:def_cut_small_pis}
    \end{equation}
In the second and third constraints, \(A_{keep}x \leq D_{keep}b\) and \(b \leq A_{keep}x\), is used to determine the variable \(b\) from the variable \(x\). Each row of \(A_{keep} \in \mathbb{B}^{m\times n}\) corresponds to some path \(P \in S_{m_i,i}\), and \(D_{keep}\in \mathbb{B}^{m\times m}\) is a diagonal matrix. Suppose the \(q\)-th row of \(A_{keep}\) corresponds to a path \(P \in S_{m_i,i}\) for \(P_{keep} = (S_{m_1,1},\ldots, S_{m_n,n})\), and \(MC_{keep}(i)\) is the set of minimum-cut edges on some path in \(S_{m_i,i}\), then the \(q\)-th row is constructed as follows,
    \begin{equation}
        (A_{keep})_{q,r} := \begin{cases}
      1, & \text{if}\ E_{cut}(r) \in P \cap MC_{keep}(i).\\
      0, & \text{otherwise.}
    \end{cases}
        \label{eq:def_cut_consecutive_pis}
    \end{equation}
The \(q\)-th diagonal entry of \(D_{keep}\) stores the total number of minimum-cut edges in the path corresponding to the \(q\)-th row of \(A_{keep}\).
    \begin{equation}
        D_{keep} := diag(A_{keep}\bbone)
        \label{eq:max}
    \end{equation}
These two constraints ensure that for some \(q \leq n\), \(b_q = 1\) iff at least one minimum-cut edge on the path corresponding to the \(q\)-th row of \(A_{keep}\) is constrained, and \(b_q = 0\) iff none of the minimum-cut edges on the path corresponding to the \(q\)-th row of \(A_{keep}\) are constrained.
\\
The fourth and fifth constraints, \(D_ff \leq A_f(\bbone - b)\) and \(f \geq A_f(\bbone -b) - D_f\bbone + \bbone\), determine the flow value for a given set of sequence flow paths, \(S_f\). Suppose the \(q\)-th row of the matrix \(A_f \in \mathbb{B}^{l\times m}\) corresponds to some sequence flow path \(P = (P_1, \ldots, P_n)\in S_f\). Let \(R = (r_1, \ldots, r_n)\) denote the indices of the paths \(P_1,\ldots, P_n\) according to the ordering of the paths constituting all \(S_{m_i,i}\) that is consistent with the construction of \(A_{keep}\) and \(D_{keep}\). Then, the \(q\)-th row of \(A_f\) is defined as follows,
\begin{equation}
        (A_f)_{q,r} := \begin{cases}
      1, & \text{if}\ r = r_i \text{ for some } 1\leq i \leq n.\\
      0, & \text{otherwise.}
    \end{cases}
        \label{eq:Af}
    \end{equation}
The \(q\)-th diagonal entry of matrix \(D_f \in \mathbb{B}^{l\times l}\) stores the total number of ones in the \(q\)-th row of \(A_f\), which is always \(n\).
    \begin{equation}
        D_f := diag(A_f\bbone).
        \label{eq:Df}
    \end{equation}
\\
The first constraint of the ILP, \(A_{cut}x\geq \bbone\), enforces the requirement that each path in \(P\in  P_{cut}\) is constrained. Each row of \(A_{cut}\) corresponds to a path \(P \in P_{cut}\). The \(q\)-th row of \(A_{cut}\) is constructed as follows,
    \begin{equation}
        (A_{cut})_{q,r} = \begin{cases}
      1 & \text{if}\ E_{cut}(r) \in P = P_{cut}(q) \\
      0 & \text{otherwise.}
    \end{cases}
        \label{eq:def_cut_small_pis}
    \end{equation}
In the second and third constraints, \(A_{keep}x \leq D_{keep}b\) and \(b \leq A_{keep}x\), is used to determine the variable \(b\) from the variable \(x\). Each row of \(A_{keep} \in \mathbb{B}^{m\times n}\) corresponds to some path \(P \in S_{m_i,i}\), and \(D_{keep}\in \mathbb{B}^{m\times m}\) is a diagonal matrix. Suppose the \(q\)-th row of \(A_{keep}\) corresponds to a path \(P \in S_{m_i,i}\) for \(P_{keep} = (S_{m_1,1},\ldots, S_{m_n,n})\), and \(MC_{keep}(i)\) is the set of minimum-cut edges on some path in \(S_{m_i,i}\), then the \(q\)-th row is constructed as follows,
    \begin{equation}
        (A_{keep})_{q,r} := \begin{cases}
      1, & \text{if}\ E_{cut}(r) \in P \cap MC_{keep}(i).\\
      0, & \text{otherwise.}
    \end{cases}
        \label{eq:def_cut_consecutive_pis}
    \end{equation}
The \(q\)-th diagonal entry of \(D_{keep}\) stores the total number of minimum-cut edges in the path corresponding to the \(q\)-th row of \(A_{keep}\).
    \begin{equation}
        D_{keep} := diag(A_{keep}\bbone)
        \label{eq:max}
    \end{equation}
These two constraints ensure that for some \(q \leq n\), \(b_q = 1\) iff at least one minimum-cut edge on the path corresponding to the \(q\)-th row of \(A_{keep}\) is constrained, and \(b_q = 0\) iff none of the minimum-cut edges on the path corresponding to the \(q\)-th row of \(A_{keep}\) are constrained.
\\
The fourth and fifth constraints, \(D_ff \leq A_f(\bbone - b)\) and \(f \geq A_f(\bbone -b) - D_f\bbone + \bbone\), determine the flow value for a given set of sequence flow paths, \(S_f\). Suppose the \(q\)-th row of the matrix \(A_f \in \mathbb{B}^{l\times m}\) corresponds to some sequence flow path \(P = (P_1, \ldots, P_n)\in S_f\). Let \(R = (r_1, \ldots, r_n)\) denote the indices of the paths \(P_1,\ldots, P_n\) according to the ordering of the paths constituting all \(S_{m_i,i}\) that is consistent with the construction of \(A_{keep}\) and \(D_{keep}\). Then, the \(q\)-th row of \(A_f\) is defined as follows,
\begin{equation}
        (A_f)_{q,r} := \begin{cases}
      1, & \text{if}\ r = r_i \text{ for some } 1\leq i \leq n.\\
      0, & \text{otherwise.}
    \end{cases}
        \label{eq:Af}
    \end{equation}
The \(q\)-th diagonal entry of matrix \(D_f \in \mathbb{B}^{l\times l}\) stores the total number of ones in the \(q\)-th row of \(A_f\), which is always \(n\).
    \begin{equation}
        D_f := diag(A_f\bbone).
        \label{eq:Df}
    \end{equation}
\\
The fourth constraint ensures that if any of the constituent paths, \(P_1, \ldots, P_n\), in the \(q\)-th sequence flow path \(P = (P_1, \ldots, P_n)\in S_f\) (for \(1\leq q\leq l\)), is constrained, then the flow value, \(f_q = 0\). The last constraint ensures that if none of the constituent paths, \(P_1, \ldots, P_n\), in the \(q\)-th sequence flow path \(P = (P_1, \ldots, P_n)\in S_f\) (for \(1\leq q\leq l\)), are constrained, then the flow value, \(f_q = 1\).

\subsection{Proofs}
\label{sec:appendix_proofs}
\textbf{Proof of Lemma~\ref{lemma:min_cut_edges}}
\begin{proof}
A path \(P_{cut} \in \mP_{cut}\) can be constrained by removing at least one of its constituent edges. The number of edges of \(P_{cut}\) that are not in some path \(P \in \mP\) is non-zero, since otherwise it would imply that \(P_{cut,i}\in \mP\), and would not need to be constrained. The set \(C\) can simply be chosen by selecting one or more edges on every \(P_{cut} \in \mP_{cut}\) that are not a part of some path in \(\mP\). \qed
\end{proof}
\textbf{Proof of Lemma~\ref{lemma:test_graph_test_spec}}
\begin{proof}
From Assumption~\ref{assump:n_goal_nodes}, there is only one node in \(G'\) for each proposition in characterizing the test specification ~(\ref{eq:test_spec}), and node satisfying proposition \(p_i\) is labeled as \(v_i\). For every \(i\in\{1,\cdots, n\}\), \(v_i\) is the only state in test graph \(G'\) that is successor to all states \(v\) on paths \(Paths(v_{j<i}, v_{n+1})\) for which \(d_{G'}(v, v_{n+1}) = d_{G'}(v_i, v_{n+1})+1\). This is true by construction of the ILP constraints; all \(Paths(p_{j<i},g)\) on the test graph \(G'\) pass through \(v_i\). \\
Let \(\sigma\) denote the test run of the agent starting at \(v_1\). We define a metric on the test graph \(G'\): \(m_t := \min_t d_{G'}(\sigma_t, v_{n+1})\) to be the closest distance to node \(v_{n+1}\) in the first \(t\) steps of the test run. Note three properties of this metric \(m_t\): (a) \(m_t \geq 0\), (b) \(m_t\) decreases: \(m_{t+1} := \min \{\sigma_{t+1}, m_t\} \leq m_t\), and (c) there exists a successor \(q_{t+1}\) to \(\sigma_t=q_t\) on \(G'\) such that \(d_{G'}(q_{t+1}, v_{n+1}) = d_{G'}(q_{t}, v_{n+1}) - 1\) that decreases \(m_t\). The metric \(m_t\) starts at \(m_0 \geq d_{G'}(v_1,v_{n+1})\) and decreases to \(0\) at the end of the test run. Thus, we can observe that \(\sigma \models \lozenge(p_1 \wedge \lozenge(p_2 \cdots \wedge \lozenge p_{n+1})) \wedge_{i=1}^{n} (\neg p_{i+1} \mathsf{U}\, p_i) \iff \sigma \models \lozenge v_{n+1}\) . \qed
\end{proof}
\textbf{Proof of proposition~\ref{prop:min_constr}}
\begin{proof}
By construction, the inputs to the ILP~\eqref{eq:ILP_Problem} are constructed based on a maximal set of sequence flow paths from \(v_1\) to \(v_n\). By Lemma~\ref{lemma:min_cut_edges}, at each iteration of the ILP~\eqref{eq:ILP_Problem} from which constraint edges are chosen, the maximum sequence flow value does not decrease at each iteration. Since there are a finite number of edges
there are a finite number of iterations until test graph is found. Therefore, the Algorithm~\ref{alg:StaticConstraints} returns a minimally constrained test graph. \qed
\end{proof}
\subsection{Complexity of Subroutines in Algorithm~\ref{alg:StaticConstraints}}
\label{sec:appendix_complexity}
Since Find-Cut-Paths is determining a set of augmenting paths for a single source-sink flow, it has a complexity of Edmonds-Karp algorithm, \(O(|V||E|^2)\) time for graph \(G=(V,E)\)~\cite{cormen2009introduction}. The complexity of Min-cut-Edges is \(O(|V||E|^3)\) time since it runs a max-flow algorithm for each edge in the worst-case. The main computational bottleneck is in the Sequence-Flows subroutine, which constructs sets of augmenting flows by computing combinations of all simple paths and all shortest paths. In the worst-case, enumerating all simple paths between two nodes is \(O|V!|\), and enumerating all shortest paths is slightly better in several cases. 
\end{document}